\documentclass[conference]{IEEEtran}
\usepackage{cite}
\usepackage{amsmath,amssymb,amsfonts}
\usepackage{amsthm}
\usepackage{mathtools}
\usepackage{graphicx}
\usepackage{booktabs}
\usepackage{array}
\usepackage{multirow}
\usepackage{makecell}
\usepackage{url}
\usepackage{listings}
\usepackage{xcolor}
\usepackage{pifont}
\usepackage{tikz}
\usepackage{pgfplots}
\pgfplotsset{compat=1.16}
\usetikzlibrary{arrows.meta,positioning,shapes.geometric,shapes.misc,fit,calc,backgrounds}

\newtheorem{theorem}{Theorem}
\newtheorem{lemma}{Lemma}
\newtheorem{corollary}{Corollary}
\theoremstyle{definition}
\newtheorem{definition}{Definition}

\definecolor{appfill}{RGB}{223,236,248}
\definecolor{appline}{RGB}{30,90,146}
\definecolor{ctrfill}{RGB}{221,242,236}
\definecolor{ctrline}{RGB}{18,120,104}
\definecolor{ledfill}{RGB}{237,238,241}
\definecolor{ledline}{RGB}{88,90,96}
\definecolor{acc}{RGB}{176,42,42}
\definecolor{gold}{RGB}{181,138,0}
\definecolor{ok}{RGB}{26,127,55}
\definecolor{barfill}{RGB}{30,90,146}
\definecolor{barfill2}{RGB}{176,42,42}
\definecolor{linkblue}{RGB}{0,63,175}

\newcommand{\cmark}{\textcolor{ok}{\ding{51}}}
\newcommand{\xmark}{\textcolor{acc}{\ding{55}}}
\newcommand{\pmark}{\textcolor{gold}{\ding{115}}}

\lstdefinelanguage{Solidity}{
  keywords={contract, function, modifier, mapping, address, uint256, uint64, uint32, string, bytes32, bool, struct, enum, require, emit, event, public, external, internal, view, returns, calldata, memory, storage, if, else, delete, constructor, indexed, true, false},
  sensitive=true,
  comment=[l]{//},
  morestring=[b]"
}
\usepackage[
  colorlinks=true,
  linkcolor=linkblue,
  citecolor=linkblue,
  urlcolor=linkblue,
  breaklinks=true,
  bookmarksnumbered=true,
  pdftitle={Making Duplicate Reimbursement Unrepresentable: A Verified Ethereum E-Invoice System for Humans and AI Agents},
  pdfauthor={Jia Cai},
  pdfkeywords={blockchain; Ethereum; smart contract; electronic invoice; non-fungible token; formal model; autonomous agent; tax administration}
]{hyperref}

\newcommand{\artifacturl}{https://github.com/jeffjiacai/ethereum-e-invoice-system}

\begin{document}

\title{Making Duplicate Reimbursement Unrepresentable:\\A Verified Ethereum E-Invoice System\\for Humans and AI Agents}

\author{\IEEEauthorblockN{Jia Cai}
\IEEEauthorblockA{College of Engineering and Computing \\
George Mason UNiversity\\
Fairfax, VA, USA\\
jcai8@gmu.edu}}

\maketitle

\begin{abstract}
Electronic invoices are replacing paper invoices worldwide, but today's centralized e-invoice architectures leave three problems unsolved on the consumption side: an e-invoice can be printed and submitted for reimbursement repeatedly, its authenticity is hard for recipients to verify, and invoice data is siloed at a central authority that becomes both a performance bottleneck and a single point of failure. This paper presents the design, formal analysis, and implementation of a complete blockchain-based electronic invoice system on Ethereum. We give a formal model of the invoice lifecycle as a guarded labeled transition system and prove, under standard cryptographic and consensus assumptions, that the system guarantees (i) reimbursement uniqueness---an invoice can be reimbursed at most once, even across mutually distrusting organizations; (ii) face integrity---any invoice that passes verification equals the recorded one unless keccak256 second-preimage resistance is broken; and (iii) authorization soundness for every lifecycle operation. The core state-machine invariants are additionally \emph{machine-checked} with the Solidity SMTChecker, which proves them inductively over all reachable transaction sequences. The design models each invoice as a non-fungible, \emph{non-tradable} token whose ownership and state change only through five lifecycle subsystems, and a lock-based reimbursement protocol makes duplicate reimbursement unrepresentable rather than merely detectable. We implement the design as a Solidity~0.8 smart contract with a four-role web application and evaluate it on a private Ethereum network: issuing an invoice costs 646{,}773~gas, the complete reimbursement protocol costs under 135{,}000~gas, every core operation is $O(1)$ in the number of invoices, and a single development node sustains 137~invoice issuances per second. Finally, we show that the verified contract doubles as a \emph{safety envelope for autonomous AI agents}: an LLM-based reimbursement agent operating a registered account is exactly the adversary of our threat model, so agent safety follows as a corollary of the proved theorems, and in an end-to-end case study every unsafe agent action --- duplicate, over-limit, or forged-receipt claims, even with the agent's own policy layer bypassed --- is rejected, ultimately by the contract itself. The full implementation, agent runtime, test suite, and benchmarks are open source. The results indicate that lifecycle-complete, formally-grounded invoice management on a blockchain is practical and directly eliminates the duplicate-reimbursement and verification pain points of centralized designs.
\end{abstract}


\section{Introduction}
Invoices are the primary written evidence of commercial transactions. They anchor bookkeeping and expense reimbursement inside firms, value-added tax (VAT) collection by tax authorities, audit by regulators, and evidence in litigation. Their digitization is accelerating globally: China has operated a nationwide VAT e-invoice system since 2015~\cite{sta2015}, and the European Union's ``VAT in the Digital Age'' package, adopted in March 2025, mandates structured electronic invoicing for intra-EU B2B transactions by 2030, projecting fraud reductions of up to €11~billion per year~\cite{vida2025,ecvida2025}.

Existing e-invoice systems are centralized: a tax authority (or its service provider) issues, records, and validates all invoices. This architecture digitizes the \emph{issuing} side effectively but neglects the \emph{consumption} side, where three pain points persist.

\textbf{Duplicate reimbursement.} An electronic invoice is a file that can be printed or forwarded any number of times. An employee can submit the same invoice for reimbursement at two employers, or twice at the same employer through different channels. Finance departments today defend against this with manually maintained registers of already-reimbursed invoice numbers---a process that is error-prone, unauditable, and does not work across organizations at all.

\textbf{Costly verification.} Paper invoices carried physical anti-forgery features; a printed e-invoice carries none. Recipients must query the authority's central platform to verify authenticity, and in practice most do not, so forged and altered e-invoices circulate. Some firms respond by refusing e-invoices above certain amounts, undermining adoption.

\textbf{Central bottleneck and data silo.} All invoice data converges on the authority's platform, which must be provisioned for peak nationwide load, becomes a single point of failure, and gives downstream users (buyers, auditors, banks, courts) no direct, tamper-evident access to invoice state.

These pain points are precisely the trust and state-sharing problems that blockchains address~\cite{zheng2017,christidis2016}. A blockchain provides a replicated, append-only ledger whose state transitions are validated by consensus rather than by a single operator; a smart contract platform such as Ethereum~\cite{buterin2014,wood2014} additionally lets the invoice lifecycle itself be encoded as executable rules that no participant can bypass. If every invoice is a unique on-chain asset whose reimbursement is a one-way state transition, duplicate reimbursement is not merely detectable---it is \emph{unrepresentable}.

This paper designs, formally analyzes, implements, and evaluates such a system. Concretely, we make six contributions:

\begin{enumerate}
\item \textbf{Requirements and architecture.} We analyze the business requirements of electronic invoicing from the perspectives of the tax authority, the invoice issuer (seller), the recipient (buyer), and third-party verifiers, and derive a five-subsystem architecture covering the complete invoice lifecycle (Section~\ref{sec:arch}).
\item \textbf{Formal model.} We formalize the system as a guarded labeled transition system over an explicit global state, with an invoice lifecycle automaton and a complete table of transition rules (Section~\ref{sec:model}).
\item \textbf{Security proofs, machine-checked.} Under an explicit threat model, we prove reimbursement uniqueness, face integrity, red-flush value conservation, and authorization soundness, and mechanically verify the core state-machine invariants with the Solidity SMTChecker (Section~\ref{sec:security}).
\item \textbf{Complete implementation.} We realize the design as a single auditable Solidity~0.8.24 contract that models invoices as non-fungible, \emph{non-tradable} tokens, together with a four-role web application, released as open source (Sections~\ref{sec:design}--\ref{sec:impl}).
\item \textbf{Evaluation.} We evaluate functional correctness (23 unit tests spanning all subsystems), per-operation gas cost, asymptotic complexity, latency, and throughput on a private Ethereum network (Section~\ref{sec:eval}).
\item \textbf{Agentic case study.} We build an autonomous reimbursement agent (hybrid deterministic/LLM reasoning) and a ledger-audit agent on top of the contract, show that agent safety is a corollary of the proved theorems --- the agent is a registered account, i.e., the modeled adversary --- and demonstrate end-to-end that unsafe agent actions are rejected even when the agent's own policy layer is bypassed (Section~\ref{sec:agents}).
\end{enumerate}

The rest of the paper is organized as follows. Section~\ref{sec:related} surveys related work and positions our contribution. Section~\ref{sec:prelim} gives preliminaries and the threat model. Section~\ref{sec:model} presents the formal model. Sections~\ref{sec:arch}--\ref{sec:design} give the architecture and contract design; Section~\ref{sec:security} proves the security properties; Sections~\ref{sec:impl}--\ref{sec:eval} cover implementation and evaluation; Section~\ref{sec:agents} presents the autonomous-agent layer and its case study; Sections~\ref{sec:discussion}--\ref{sec:conclusion} discuss limitations and conclude.

\section{Related Work}
\label{sec:related}

\textbf{Blockchain platforms.} Bitcoin introduced the replicated append-only transaction ledger secured by proof of work~\cite{nakamoto2008}. Ethereum generalized the model with a Turing-complete virtual machine and persistent contract state~\cite{buterin2014,wood2014}, enabling applications beyond currency transfer~\cite{christidis2016}. Consortium deployments typically replace proof of work with committee-based consensus such as PBFT~\cite{castro1999}, which suits closed ecosystems like taxation, where participants are identified and permissioned. Our contract layer is consensus-agnostic and runs unmodified on either configuration.

\textbf{Blockchain for tax and invoicing.} Hyv\"arinen et al.\ prototyped a blockchain service that eliminates dividend-tax refund fraud in Denmark by making refund state explicit and shared~\cite{hyvarinen2017}---the same ``make fraud unrepresentable'' principle we apply to reimbursement. Fatz et al.\ proposed decentralized validation of VAT processes to achieve tax compliance by design~\cite{fatz2019}. Zhang and Liu analyzed the data-sharing characteristics of e-invoices and sketched a blockchain e-invoice design, without a lifecycle implementation~\cite{zhang2017}. On the industrial side, the Shenzhen Municipal Taxation Bureau and Tencent launched China's first production blockchain e-invoice in August 2018~\cite{tencent2018}, issuing about six million invoices in the first year~\cite{cointelegraph2019}; the system is proprietary, and its contract-level design is not public.

\textbf{Smart-contract e-invoice and VAT systems.} Closest to our work, Nguyen et al.\ digitize invoices with an Ethereum smart contract combined with a decentralized storage network, authenticating a transaction and then computing and approving its VAT payment~\cite{nguyen2019}. More recently, Farchan implemented e-invoice (\emph{e-Faktur}) issuance and validation for Indonesia's VAT system on a permissioned Hyperledger Fabric ledger, evaluating resilience against simulated attacks~\cite{farchan2024}. Both target the \emph{issuing and reporting} side---authenticating an invoice and settling its VAT---and neither addresses the consumption-side lifecycle: neither models invoices as ownable tokens, implements red-flush credit-note correction, nor prevents duplicate \emph{reimbursement}. A parallel, largely industry-driven line of work tokenizes invoices as tradable ERC-721 assets~\cite{erc721} for \emph{invoice financing}, where the analogous hazard is double-financing the same receivable; there tradability is the whole point, the opposite of our non-tradable design for legal tax documents.

\textbf{Positioning.} Table~\ref{tab:compare} summarizes the comparison. Relative to prior academic work, which is largely conceptual, targets a single fraud scenario, or covers only invoice issuance and VAT settlement, this paper contributes a lifecycle-complete, \emph{formally-analyzed}, open-source design---from blank-invoice distribution through red-flush correction to lock-based reimbursement---with proved security properties and reproducible cost measurements. Relative to proprietary industrial systems, it provides a published, auditable contract design; we follow established Solidity security patterns~\cite{wohrer2018,delmolino2016}.

\begin{table*}[t]
\caption{Comparison with representative blockchain e-invoice / tax systems.
\cmark~= supported, \pmark~= partial or proprietary, \xmark~= not addressed.}
\label{tab:compare}
\centering
\renewcommand{\arraystretch}{1.2}
\resizebox{\textwidth}{!}{%
\begin{tabular}{@{}l l c c c c c c@{}}
\toprule
System & Platform & Invoice model & Lifecycle coverage & Digest verification & Red-flush & \makecell{Duplicate-reimburse\\prevention} & \makecell{Formal proofs /\\open source} \\
\midrule
Zhang \& Liu '17~\cite{zhang2017}      & Conceptual        & data record         & issuance (sketch)         & \pmark & \xmark & \xmark & \xmark~/~\xmark \\
Hyv\"arinen et al.\ '17~\cite{hyvarinen2017} & Ethereum      & refund claim        & dividend-tax refund       & \pmark & \xmark & \cmark~(refund) & \xmark~/~\xmark \\
Shenzhen/Tencent '18~\cite{tencent2018} & Consortium       & data record         & issue $\to$ reimburse      & \cmark & \pmark & \pmark~(closed) & \xmark~/~\xmark \\
Fatz et al.\ '19~\cite{fatz2019}       & Ethereum          & process instance    & VAT process validation    & \pmark & \xmark & \xmark & \xmark~/~\xmark \\
Nguyen et al.\ '19~\cite{nguyen2019}   & Ethereum + DSN    & data record         & issue + VAT payment       & \cmark & \xmark & \xmark & \xmark~/~\xmark \\
Farchan '24~\cite{farchan2024}         & Hyperledger Fabric& data record         & issue + reporting         & \cmark & \xmark & \xmark & \xmark~/~\xmark \\
\textbf{This work}                     & Ethereum / EVM    & non-tradable NFT    & full five-subsystem       & \cmark & \cmark & \cmark~(cross-org) & \cmark~/~\cmark \\
\bottomrule
\end{tabular}%
}
\end{table*}

\section{Preliminaries and Threat Model}
\label{sec:prelim}

\subsection{Blockchain and Smart Contracts}
Ethereum maintains a replicated state machine: externally owned accounts hold balances and issue signed transactions, and contract accounts hold code and storage executed by the Ethereum Virtual Machine (EVM)~\cite{wood2014}. A \emph{smart contract} is a program whose functions are invoked by transactions; every state-changing invocation is ordered by consensus, recorded immutably, and metered in \emph{gas}, which measures computational and storage cost independently of any token price. Contracts emit \emph{events}---indexed log entries that clients query to reconstruct a complete, tamper-evident history of an asset.

Two properties are load-bearing for invoicing. First, storage writes are permanent and globally replicated, so an on-chain invoice cannot be silently altered or deleted; corrections must themselves be recorded transactions. Second, contract code enforces its own invariants: if the contract exposes no function that reimburses an invoice twice, no participant---including the operator---can do so.

\subsection{Cryptographic Primitives}
We use the keccak256 hash function $H:\{0,1\}^*\!\to\!\{0,1\}^{256}$, modeled as second-preimage- and collision-resistant, and ECDSA signatures over secp256k1, assumed existentially unforgeable under chosen-message attack (EUF-CMA); the EVM authenticates every transaction's sender by signature. Consensus is assumed \emph{safe} (no two conflicting histories finalize) under its fault threshold---honest hash-power majority for proof of work, or fewer than $n/3$ Byzantine nodes for PBFT-class consortium consensus~\cite{castro1999}. Consequently, transactions apply to a single, totally ordered, globally agreed state.

\subsection{Threat Model}
\label{sec:threat}
The adversary $\mathcal{A}$ controls the private keys of an arbitrary set of \emph{registered} enterprise accounts and may submit any transactions in any order, adaptively. This explicitly includes \emph{faulty or adversarial autonomous agents}: software (including LLM-based agents, Section~\ref{sec:agents}) operating a registered account is, from the contract's perspective, indistinguishable from any other key holder, so every guarantee proved against $\mathcal{A}$ applies verbatim to arbitrary agent behavior, including behavior induced by prompt injection. $\mathcal{A}$ cannot forge signatures of honest accounts (EUF-CMA), cannot find collisions or second preimages of $H$, and cannot violate consensus safety. The tax bureau key is honest for \emph{admission} (registration and blank-invoice supply), reflecting its legal role, but---unlike centralized systems---is \emph{not} trusted for invoice state, which the contract and consensus enforce. $\mathcal{A}$ succeeds if it can (G1) cause some invoice to be reimbursed twice; (G2) present an invoice face that verifies yet differs from the recorded one; or (G3) drive a lifecycle transition on an invoice it is not authorized to act on. Section~\ref{sec:security} proves each goal infeasible. Off-chain concerns (fictitious underlying sales, key theft, storage-layer confidentiality) are out of scope for these guarantees and discussed in Section~\ref{sec:discussion}.

\section{Formal Model}
\label{sec:model}

We model the contract as a deterministic transition system over a global state. Table~\ref{tab:notation} lists the notation.

\begin{table}[t]
\caption{Notation for the formal model.}
\label{tab:notation}
\centering
\renewcommand{\arraystretch}{1.15}
\begin{tabular}{@{}c l@{}}
\toprule
Symbol & Meaning \\
\midrule
$\mathcal{A}$ & address space; $b\in\mathcal{A}$ the tax bureau \\
$\mathcal{N}$ & invoice-number space (nonzero) \\
$R\subseteq\mathcal{A}$ & registered enterprise accounts \\
$O:\mathcal{N}\rightharpoonup\mathcal{A}$ & invoice owner (undefined $=$ nonexistent) \\
$\mathsf{st}:\mathcal{N}\to Q$ & lifecycle status, $Q$ from Def.~\ref{def:lts} \\
$\mathsf{cr}:\mathcal{N}\to\{0,1\}$ & credit-note flag \\
$\mathsf{dc}:\mathcal{N}\to\{0,1\}$ & tax-declared flag \\
$F:\mathcal{N}\rightharpoonup \mathbb{F}$ & invoice face record (Def.~\ref{def:invoice}) \\
$D:\{0,1\}^{256}\rightharpoonup\mathcal{N}$ & digest index \\
$L:\mathcal{N}\rightharpoonup \mathcal{A}\times \mathsf{Cid}$ & reimbursement lock (locker, claim id) \\
$\mathsf{sel}(n),\mathsf{buy}(n)$ & seller / buyer address of invoice $n$ \\
$\mathsf{tot}(n)$ & tax-inclusive total of invoice $n$ (cents) \\
\bottomrule
\end{tabular}
\end{table}

\begin{definition}[Global state]
\label{def:state}
A state is a tuple $\sigma=(R,O,\mathsf{st},\mathsf{cr},\mathsf{dc},F,D,L)$. The genesis state $\sigma_0$ has $R=\varnothing$ and all partial maps empty; $\mathsf{st}(n)=\textsc{None}$ wherever $O(n)$ is undefined.
\end{definition}

\begin{definition}[Invoice face]
\label{def:invoice}
A face $F(n)\in\mathbb{F}$ is the record $\langle \mathsf{sel},\mathsf{buy}$ (each a snapshot of taxpayer id, name, bank, address)$,\ p,\ r,\ t,\ \mathsf{tot},\ \kappa,\ \mathit{items},\ \mathsf{cr},\ \tau\rangle$, where $p$ is the pre-tax amount in cents, $r$ the rate in basis points, $t=\lfloor p\cdot r/10^4\rfloor$ the tax, $\mathsf{tot}=p+t$, $\kappa$ the category code, and $\tau$ the issuance time. Its digest is $H(F(n))$ over the canonical encoding of $\langle n,\mathsf{sel}.\mathit{id},\mathsf{buy}.\mathit{id},p,r,\kappa,\mathit{items},\mathsf{cr},\tau\rangle$.
\end{definition}

\begin{definition}[Lifecycle transition system]
\label{def:lts}
The invoice lifecycle is the labeled transition system $\mathcal{L}=(Q,\Sigma,\rightarrow,\textsc{None})$ with states
$Q=\{\textsc{None},\textsc{Blank},\textsc{Issued},\textsc{Locked},\textsc{Reimbursed},\textsc{Reversed}\}$,
labels $\Sigma$ the contract operations of Table~\ref{tab:transitions}, and $\rightarrow$ the per-invoice projection of those rules (Fig.~\ref{fig:fsm}). \textsc{Reimbursed} and \textsc{Reversed} are terminal (no outgoing edges).
\end{definition}

Each operation is a \emph{guarded command} $\mathsf{op}(c,\vec{x}): \mathit{guard}(\sigma,c,\vec{x})\Rightarrow \sigma'$, where $c$ is the transaction sender; if the guard fails the transaction reverts and $\sigma$ is unchanged. Table~\ref{tab:transitions} gives the guards and effects; these are exactly the \texttt{require} clauses and state writes of the contract. An \emph{execution} is a finite sequence $\sigma_0 \xrightarrow{\mathsf{op}_1}\sigma_1\cdots$; by consensus safety it is a single total order.

\begin{table*}[t]
\caption{Guarded transition rules (per invoice $n$). $c$ is the caller. All guards are conjunctive; a failed guard reverts with no effect.}
\label{tab:transitions}
\centering
\renewcommand{\arraystretch}{1.3}
\resizebox{\textwidth}{!}{%
\begin{tabular}{@{}l l l@{}}
\toprule
Operation & Guard (precondition) & Effect (state update) \\
\midrule
$\mathsf{register}(c,a,\mathit{id})$ & $c=b \,\wedge\, a\notin R$ & $R \mathrel{+}= a$ \\
$\mathsf{mint}(c,n,a)$ & $c=b \,\wedge\, a\in R \,\wedge\, O(n){=}\bot$ & $O(n){:=}a;\ \mathsf{st}(n){:=}\textsc{Blank}$ \\
$\mathsf{issue}(c,n,v,\phi)$ & $c\in R \wedge O(n){=}c \wedge \mathsf{st}(n){=}\textsc{Blank} \wedge v\in R \wedge v\neq c \wedge \mathit{valid}(\phi)$ & $F(n){:=}\phi;\ \mathsf{st}(n){:=}\textsc{Issued};\ O(n){:=}v;\ D(H(\phi)){:=}n$ \\
$\mathsf{redFlush}(c,n,m)$ & $c\in R \wedge \mathsf{st}(n){=}\textsc{Issued} \wedge \mathsf{cr}(n){=}0 \wedge \mathsf{sel}(n){=}c \wedge O(m){=}c \wedge \mathsf{st}(m){=}\textsc{Blank}$ & \makecell[l]{$F(m){:=}F(n),\ \mathsf{cr}(m){:=}1;\ \mathsf{st}(n){:=}\textsc{Reversed};$\\$\mathsf{st}(m){:=}\textsc{Issued};\ O(m){:=}\mathsf{buy}(n);\ D(H(F(m))){:=}m$} \\
$\mathsf{declare}(c,n)$ & $c=\mathsf{sel}(n) \wedge \mathsf{st}(n)\notin\{\textsc{None},\textsc{Blank}\} \wedge \mathsf{cr}(n){=}0 \wedge \mathsf{dc}(n){=}0$ & $\mathsf{dc}(n){:=}1$ \\
$\mathsf{lock}(c,n,\mathit{cid})$ & $\mathsf{st}(n){=}\textsc{Issued} \wedge \mathsf{cr}(n){=}0 \wedge O(n){=}c \ [\wedge\ \mathsf{dc}(n){=}1]$ & $\mathsf{st}(n){:=}\textsc{Locked};\ L(n){:=}(c,\mathit{cid})$ \\
$\mathsf{reimburse}(c,n)$ & $\mathsf{st}(n){=}\textsc{Locked} \wedge L(n).\mathit{locker}{=}c$ & $\mathsf{st}(n){:=}\textsc{Reimbursed}$ \\
$\mathsf{unlock}(c,n)$ & $\mathsf{st}(n){=}\textsc{Locked} \wedge L(n).\mathit{locker}{=}c$ & $\mathsf{st}(n){:=}\textsc{Issued};\ L(n){:=}\bot$ \\
\bottomrule
\end{tabular}%
}
\end{table*}

\begin{figure}[t]
\centering
\resizebox{\columnwidth}{!}{%
\begin{tikzpicture}[
  font=\footnotesize,
  every node/.style={font=\footnotesize},
  st/.style={draw=ctrline, line width=0.8pt, fill=ctrfill, rounded corners=3pt, align=center, minimum width=1.35cm, minimum height=0.72cm},
  term/.style={draw=acc, line width=1pt, fill=acc!12, align=center, minimum width=1.5cm, minimum height=0.78cm, double, double distance=1pt, rounded corners=3pt},
  cr/.style={draw=gold, line width=0.8pt, fill=gold!12, rounded corners=3pt, align=center, minimum width=1.4cm, minimum height=0.72cm},
  lbl/.style={font=\scriptsize\itshape, inner sep=1.5pt},
  e/.style={-{Stealth[length=2mm]}, line width=0.7pt}
]
\node[st] (blank) {Blank};
\node[st, right=1.7cm of blank] (issued) {Issued};
\node[st, right=2.0cm of issued] (locked) {Locked};
\node[term, right=1.7cm of locked] (reimb) {Reimbursed};
\node[term, below=1.05cm of issued] (reversed) {Reversed};
\node[cr, below=1.05cm of locked] (credit) {Issued\\\scriptsize(credit)};
\node[left=0.7cm of blank] (start) {};

\draw[e] (start) -- node[lbl,above]{mint} (blank);
\draw[e] (blank) -- node[lbl,above]{issue} (issued);
\draw[e, transform canvas={yshift=2.2pt}] (issued) -- node[lbl,above]{lock} (locked);
\draw[e, transform canvas={yshift=-2.2pt}] (locked) -- node[lbl,below]{unlock} (issued);
\draw[e] (locked) -- node[lbl,above]{reimburse} (reimb);
\draw[e] (issued) -- node[lbl,left,align=right]{redFlush\\(original)} (reversed);
% credit-note edge is routed through a lane below the bottom row so it neither
% crosses the Reversed node nor collides with its label
\coordinate (lane) at ([yshift=-0.55cm]reversed.south);
\draw[e, rounded corners=3pt] (blank.south) -- (blank |- lane)
  -- node[lbl,below]{redFlush (credit note)} (credit |- lane) -- (credit.south);
\draw[{Stealth[length=1.6mm]}-{Stealth[length=1.6mm]}, dashed, gold, line width=0.7pt] (reversed) -- node[lbl,below]{linked} (credit);
\end{tikzpicture}%
}
\caption{Invoice lifecycle automaton $\mathcal{L}$ (Def.~\ref{def:lts}). Double-bordered red states are terminal; the gold node is a frozen credit note. Tax declaration is an orthogonal one-way flag $\mathsf{dc}$ and is omitted for clarity.}
\label{fig:fsm}
\end{figure}
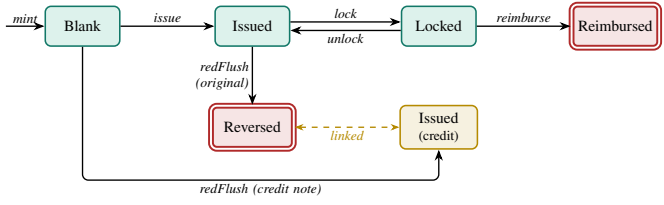

\section{Requirements and System Architecture}
\label{sec:arch}

\subsection{Actors and Functional Requirements}
Four actors participate in the invoice lifecycle:
\begin{itemize}
\item \textbf{Tax bureau}: the sole authority that admits enterprises and creates blank invoices; monitors all invoice state.
\item \textbf{Seller} (issuer): applies for blank invoices, issues invoices for real transactions, corrects erroneous invoices, declares output tax.
\item \textbf{Buyer} (recipient): receives invoices, verifies authenticity, reimburses expenses against invoices.
\item \textbf{Third-party verifier}: auditors, courts, banks, and other parties who must verify invoices without trusting seller or buyer.
\end{itemize}

From the pain points in Section~\ref{sec:prelim} and current invoice regulations we derive five functional subsystems, which together cover the complete lifecycle: (1)~\emph{application and distribution}---only the bureau may create invoices; enterprises apply and receive numbered ranges of blanks, and the bureau may grant fewer than requested; (2)~\emph{issuance and circulation}---filling a blank's face and delivering it to the buyer is one atomic action, and the face is immutable once issued; (3)~\emph{void and red-flush}---erroneous invoices cannot be edited or deleted, so the original seller issues a linked negative credit invoice and the original is marked reversed; (4)~\emph{query and verification}---any authorized party queries invoice state by number and verifies a face's authenticity from its content digest; (5)~\emph{declaration and reimbursement}---the seller declares output tax exactly once, and the buyer reimburses an invoice at most once, with an explicit failure-recovery path. Non-functional requirements include permissioned participation, auditability, throughput adequate for enterprise-scale invoicing, and no trusted intermediary on the consumption side.

\subsection{Is a Blockchain Warranted?}
Applying the standard applicability test---shared state, multiple writers, mutual distrust, and no perfect trusted third party---electronic invoicing qualifies on all four counts. Invoice state must be shared among bureau, seller, buyer, and verifiers; all of them write state at different lifecycle stages; sellers and buyers have adversarial incentives; and the existing trusted third party, the central platform, is exactly the bottleneck and silo we seek to remove. Because participants are identified enterprises under a regulator, a \emph{consortium} deployment is the natural fit, with the public-chain design retained as the more adversarial baseline.

\subsection{Architecture}
Fig.~\ref{fig:arch} shows the architecture. All lifecycle rules reside in a single smart contract (the trust boundary); role dashboards are thin clients that submit signed transactions and reconstruct history from events. No application server holds authoritative state.

\begin{figure}[t]
\centering
\resizebox{\columnwidth}{!}{%
\begin{tikzpicture}[
  font=\footnotesize,
  actor/.style={draw=appline, fill=appfill, rounded corners=2pt, align=center, minimum width=1.55cm, minimum height=0.8cm, line width=0.7pt},
  sub/.style={draw=ctrline, fill=ctrfill, rounded corners=2pt, align=center, minimum width=1.5cm, minimum height=0.85cm, line width=0.7pt},
  core/.style={draw=ctrline, fill=ctrfill!55, rounded corners=2pt, align=center, minimum width=8.2cm, minimum height=0.6cm, line width=0.7pt},
  chain/.style={draw=ledline, fill=ledfill, rounded corners=2pt, align=center, minimum width=8.4cm, minimum height=0.66cm, line width=0.7pt},
  lyr/.style={rounded corners=3pt, inner sep=6pt, line width=0.7pt},
  caplbl/.style={font=\scriptsize\bfseries, anchor=north west}
]
% application layer
\node[actor] (bureau) {\ding{110}~Tax\\bureau};
\node[actor, right=2mm of bureau] (seller) {\ding{110}~Seller};
\node[actor, right=2mm of seller] (buyer) {\ding{110}~Buyer};
\node[actor, right=2mm of buyer] (verifier) {\ding{110}~Verifier};
\begin{scope}[on background layer]
\node[lyr, draw=appline, fill=appfill!35, fit=(bureau)(verifier), inner ysep=15pt] (appL) {};
\end{scope}
\node[caplbl, text=appline] at (appL.north west) {\,Application layer — React\,/\,ethers.js};

% contract layer
\node[sub, below=13mm of bureau] (s1) {1~Apply\\\&\,distrib.};
\node[sub, right=1.2mm of s1] (s2) {2~Issue\\\&\,circ.};
\node[sub, right=1.2mm of s2] (s3) {3~Void\\\&\,red-flush};
\node[sub, right=1.2mm of s3] (s4) {4~Query\\\&\,verify};
\node[sub, right=1.2mm of s4] (s5) {5~Declare\\\&\,reimb.};
\node[core, below=2mm of s3] (registry) {registry~\textbullet~ownership~\textbullet~lifecycle FSM~\textbullet~digest index~\textbullet~events};
\begin{scope}[on background layer]
\node[lyr, draw=ctrline, fill=ctrfill!35, fit=(s1)(s5)(registry), inner ysep=15pt] (ctrL) {};
\end{scope}
\node[caplbl, text=ctrline] at (ctrL.north west) {\,Contract layer — \texttt{EInvoice.sol} on the EVM};

% ledger layer
\node[chain, below=13mm of registry] (chain) {Ethereum: consortium (PBFT-style, permissioned) \textbar\ private \textbar\ public};
\begin{scope}[on background layer]
\node[lyr, draw=ledline, fill=ledfill!55, fit=(chain), inner ysep=14pt] (ledL) {};
\end{scope}
\node[caplbl, text=ledline] at (ledL.north west) {\,Ledger layer};

\draw[-{Stealth[length=2mm]}, line width=0.8pt] (appL.south) -- node[right, font=\scriptsize, xshift=3mm]{signed txns / event queries} (ctrL.north);
\draw[-{Stealth[length=2mm]}, line width=0.8pt] (ctrL.south) -- node[right, font=\scriptsize]{consensus, replication} (ledL.north);
\end{tikzpicture}%
}
\caption{System architecture. The smart contract is the sole trust boundary; dashboards are stateless clients.}
\label{fig:arch}
\end{figure}
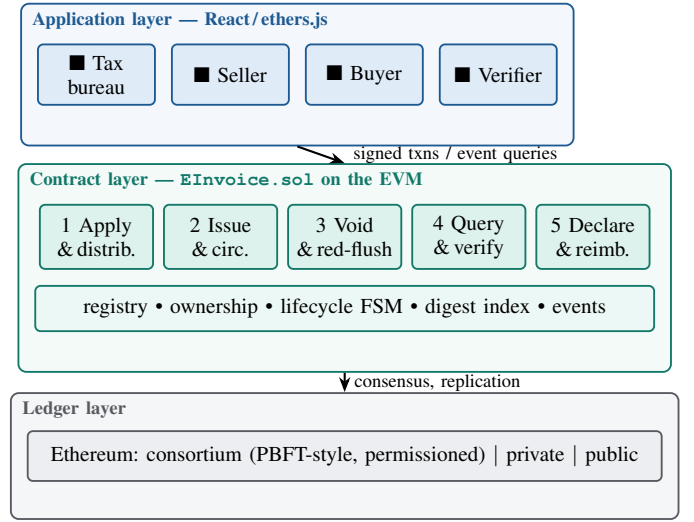

\section{Smart Contract Design}
\label{sec:design}

\subsection{Invoices as Non-Tradable, Non-Fungible Tokens}
Each invoice is unique, indivisible, and identified by its number, which naturally suggests the non-fungible token model of ERC-721~\cite{erc721}: the contract maintains \texttt{invoiceOwner} (number $\to$ holder), per-holder invoice sets, and mint/transfer primitives. We deliberately deviate from ERC-721 in one crucial respect: there is \emph{no public transfer function}. Invoices are legal documents, not tradable assets; ownership changes only as a side effect of lifecycle operations (distribution moves a blank from bureau to seller; issuance moves the invoice from seller to buyer). This closes an entire class of misuse (invoice resale, gray-market circulation) at the type-system level.

\subsection{Data Model}
The contract stores, per invoice, the fields of Def.~\ref{def:invoice}: the two parties (address plus a \emph{snapshot} of taxpayer id, legal name, bank information, and registered address, copied from the enterprise registry at issuance time so the face stays immutable even if the registry later changes), the amounts, the category code and item description, the lifecycle status, the credit flag, a bidirectional linkage field for red-flush, the issuance timestamp, and the face digest. An enterprise registry, writable only by the bureau, binds each participating address to its taxpayer identity; every lifecycle operation checks registration, giving permissioned semantics even on a public chain. Monetary values are unsigned integers (cents); credit invoices carry positive magnitudes plus the \texttt{isCredit} flag rather than signed values, avoiding sign-handling errors. Tax is computed on chain as $\lfloor \mathit{preTax}\cdot\mathit{rateBps}/10^4\rfloor$, so an arithmetically inconsistent face cannot exist.

\subsection{Subsystem 1: Application and Distribution}
An enterprise calls \texttt{applyForInvoices(count)}; the bureau either rejects or calls \texttt{approveApplication(id,\allowbreak startId, granted)} with $\mathit{granted}\le\mathit{count}$, minting the range $[\mathit{startId},\mathit{startId}{+}\mathit{granted})$ of blanks directly to the applicant. Uniqueness is enforced at mint time. Because distribution is a ledger transition rather than a portal download, the bureau needs no high-availability distribution platform, and every outstanding blank is publicly attributable to its holder.

\subsection{Subsystem 2: Issuance and Circulation}
\texttt{issueInvoice} (Listing~\ref{lst:issue}) checks that the caller holds the blank, that both parties are registered and distinct, and that amounts are valid; it then snapshots both parties, computes tax, stores the face, records the face digest, and transfers ownership to the buyer---one atomic transaction. Delivery and recording are inseparable: the buyer's very possession of the invoice implies a validated on-chain record.

\begin{lstlisting}[caption={Issuance (abridged).},label={lst:issue},float=t]
function issueInvoice(uint256 id, address buyer,
    uint256 preTax, uint256 rateBps,
    uint256 category, string calldata items)
    external onlyRegistered {
  require(invoiceOwner[id] == msg.sender);
  require(invoices[id].status == Status.Blank);
  require(enterprises[buyer].registered
          && buyer != msg.sender);
  require(preTax > 0 && rateBps <= 10000);
  _fillFace(invoices[id], msg.sender, buyer,
            preTax, rateBps, category, items, false);
  _transfer(msg.sender, buyer, id);
  emit InvoiceIssued(id, msg.sender, buyer,
       invoices[id].totalAmount,
       invoices[id].contentHash);
}
\end{lstlisting}

\subsection{Subsystem 3: Void and Red-Flush}
Ledger immutability forbids editing an issued invoice, so correction is itself a recorded transaction. \texttt{redFlush(originalId, blankId)} (Fig.~\ref{fig:redflush}) may be called only by the original seller, only on an invoice still \emph{Issued}, and only using a blank the seller holds; it copies the face onto the credit note with \texttt{isCredit} set, links both invoices bidirectionally, marks the original \emph{Reversed}, and delivers the credit note to the original buyer. Both the error and its correction remain permanently visible, eliminating the ``void and reprint'' fraud pattern of centralized systems.

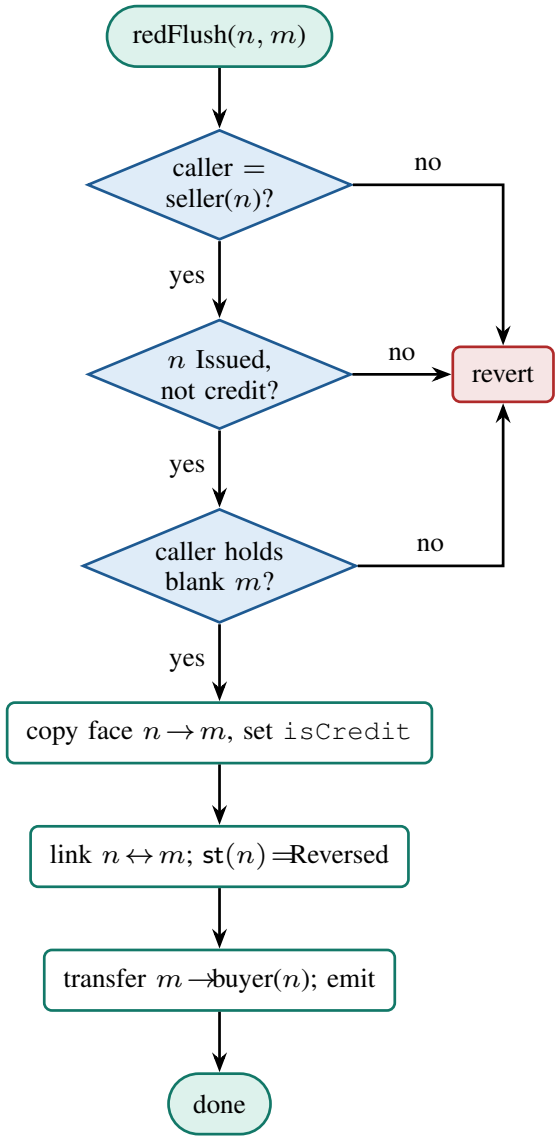
\begin{figure}[t]
\centering
\resizebox{0.82\columnwidth}{!}{%
\begin{tikzpicture}[
  font=\scriptsize, node distance=5.5mm,
  io/.style={draw=ctrline, fill=ctrfill, rounded rectangle, align=center, minimum height=0.55cm, inner xsep=6pt, line width=0.7pt},
  dec/.style={draw=appline, fill=appfill, diamond, aspect=2.4, align=center, inner sep=1pt, line width=0.7pt},
  act/.style={draw=ctrline, fill=white, rounded corners=2pt, align=center, minimum height=0.55cm, inner xsep=5pt, line width=0.7pt},
  bad/.style={draw=acc, fill=acc!12, rounded corners=2pt, align=center, minimum height=0.5cm, inner xsep=5pt, line width=0.7pt},
  e/.style={-{Stealth[length=1.8mm]}, line width=0.7pt}
]
\node[io] (s) {redFlush($n,m$)};
\node[dec, below=of s] (d1) {caller $=$\\seller($n$)?};
\node[dec, below=7mm of d1] (d2) {$n$ Issued,\\not credit?};
\node[dec, below=7mm of d2] (d3) {caller holds\\blank $m$?};
\node[act, below=7mm of d3] (a1) {copy face $n\!\to\!m$, set \texttt{isCredit}};
\node[act, below=of a1] (a2) {link $n\!\leftrightarrow\!m$; $\mathsf{st}(n)\!=\!$Reversed};
\node[act, below=of a2] (a3) {transfer $m\!\to\!$buyer($n$); emit};
\node[io, below=of a3] (end) {done};
\node[bad, right=9mm of d2] (rev) {revert};

\draw[e] (s) -- (d1);
\draw[e] (d1) -- node[left]{yes} (d2);
\draw[e] (d2) -- node[left]{yes} (d3);
\draw[e] (d3) -- node[left]{yes} (a1);
\draw[e] (a1) -- (a2);
\draw[e] (a2) -- (a3);
\draw[e] (a3) -- (end);
\draw[e] (d1) -| node[above,pos=0.25]{no} (rev);
\draw[e] (d2) -- node[above]{no} (rev);
\draw[e] (d3) -| node[above,pos=0.25]{no} (rev);
\end{tikzpicture}%
}
\caption{Red-flush control flow. Any failed guard reverts, leaving both invoices unchanged.}
\label{fig:redflush}
\end{figure}

\subsection{Subsystem 4: Query and Verification}
At issuance the contract computes $h=H(\cdot)$ over the canonical face encoding (Def.~\ref{def:invoice}) and stores $h\mapsto\mathit{id}$ in a digest index. A verifier holding a purported face recomputes $h'$ locally and calls the read-only \texttt{verifyByHash}$(h')$: a match proves the face is exactly what the issuing transaction recorded; any alteration of any field yields a digest absent from the index (Theorem~\ref{thm:integrity}). Verification requires no interaction with---or trust in---seller, buyer, or any central platform. Standard queries and the event log complete the audit surface: the full history of an invoice is reconstructible from indexed events alone, which our verifier dashboard demonstrates.

\subsection{Subsystem 5: Declaration and Lock-Based Reimbursement}
\label{sec:reimb}
The seller's \texttt{declareTax} marks an invoice's output tax declared, exactly once, and only by the seller; a bureau policy switch can require declaration before reimbursement, encoding at contract level a rule that today exists only on paper. Reimbursement is the pain point that motivates the system, implemented as a two-phase protocol (Listing~\ref{lst:reimb}, Fig.~\ref{fig:seq}) following the principle of \emph{making the invalid unrepresentable}: \texttt{lockForReimbursement} moves an \emph{Issued}, non-credit invoice held by the caller to \emph{Locked}, recording locker and claim id; \texttt{reimburse} moves a \emph{Locked} invoice to the terminal \emph{Reimbursed}, callable only by the locker; \texttt{unlockReimbursement} lets only the locker release a failed claim back to \emph{Issued}. The lock phase exists because reimbursement is a workflow, not an instant: between claim submission and approval the invoice must be unavailable to any competing claim yet recoverable if the claim fails. Binding the lock to a claim-document id also gives auditors a direct on-chain join between invoices and expense claims.

\begin{lstlisting}[caption={Lock-based reimbursement (abridged).},label={lst:reimb},float=t]
function lockForReimbursement(uint256 id,
    string calldata claimDocId) external {
  Invoice storage inv = invoices[id];
  require(inv.status == Status.Issued);
  require(!inv.isCredit);
  require(invoiceOwner[id] == msg.sender);
  inv.status = Status.Locked;
  locks[id] = ReimbursementLock(msg.sender,
      claimDocId, uint64(block.timestamp));
  emit InvoiceLocked(id, msg.sender, claimDocId);
}

function reimburse(uint256 id) external {
  require(invoices[id].status == Status.Locked);
  require(locks[id].locker == msg.sender);
  invoices[id].status = Status.Reimbursed;
  emit InvoiceReimbursed(id, msg.sender,
      invoices[id].totalAmount);
}
\end{lstlisting}

\begin{figure}[t]
\centering
\resizebox{\columnwidth}{!}{%
\begin{tikzpicture}[
  font=\scriptsize,
  msg/.style={-{Stealth[length=1.8mm]}, line width=0.7pt},
  rev/.style={-{Stealth[length=1.8mm]}, line width=0.7pt, acc, dashed},
  ret/.style={-{Stealth[length=1.8mm]}, line width=0.7pt, dashed},
  head/.style={draw, rounded corners=2pt, minimum width=1.7cm, minimum height=0.55cm, align=center}
]
% lifelines
\def\bx{0} \def\cx{3.9} \def\ax{7.4}
\node[head, fill=appfill, draw=appline] (B) at (\bx,0) {Buyer};
\node[head, fill=ctrfill, draw=ctrline] (C) at (\cx,0) {EInvoice contract};
\node[head, fill=appfill, draw=appline] (A) at (\ax,0) {Auditor};
\foreach \x in {\bx,\cx,\ax} \draw[densely dashed, ledline] (\x,-0.32) -- (\x,-5.2);

% activation bar on contract
\fill[ctrfill] (\cx-0.12,-0.75) rectangle (\cx+0.12,-4.75);
\draw[ctrline, line width=0.5pt] (\cx-0.12,-0.75) rectangle (\cx+0.12,-4.75);

\draw[msg] (\bx,-0.9) -- node[above, font=\scriptsize]{lock($n$, claim)} (\cx-0.12,-0.9);
\node[right, font=\scriptsize\itshape, text=ctrline] at (\cx+0.16,-1.2) {Issued $\to$ Locked};
\draw[msg] (\bx,-1.75) -- node[above, font=\scriptsize]{lock($n$, claim$'$)  \textit{(duplicate)}} (\cx-0.12,-1.75);
\draw[rev] (\cx-0.12,-2.15) -- node[above, font=\scriptsize, text=acc]{revert: ``not available''} (\bx,-2.15);
\draw[msg] (\bx,-2.95) -- node[above, font=\scriptsize]{reimburse($n$)} (\cx-0.12,-2.95);
\node[right, font=\scriptsize\itshape, text=ctrline] at (\cx+0.16,-3.25) {Locked $\to$ Reimbursed};
\draw[ret] (\cx-0.12,-3.6) -- node[above, font=\scriptsize]{emit InvoiceReimbursed} (\bx,-3.6);
\draw[msg] (\ax,-4.15) -- node[above, font=\scriptsize]{verifyByHash / query} (\cx+0.12,-4.15);
\draw[ret] (\cx+0.12,-4.55) -- node[above, font=\scriptsize]{authentic + audit trail} (\ax,-4.55);
\end{tikzpicture}%
}
\caption{Reimbursement sequence. The duplicate lock is rejected by the contract; any third party can independently verify and audit.}
\label{fig:seq}
\end{figure}
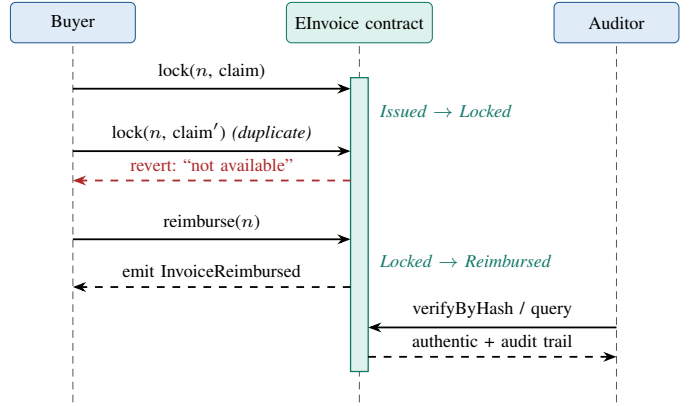

\section{Security Analysis}
\label{sec:security}
We prove that the goals of the threat model (Section~\ref{sec:threat}) are infeasible. Throughout, ``for invoice $n$'' quantifies over any single number, and executions are the totally ordered sequences of Def.~\ref{def:lts}. We first record an invariant.

\begin{lemma}[Status monotonicity]
\label{lem:mono}
In any execution, once $\mathsf{st}(n)=\textsc{Reimbursed}$ or $\mathsf{st}(n)=\textsc{Reversed}$, no subsequent transition changes $\mathsf{st}(n)$.
\end{lemma}
\begin{proof}
By inspection of Table~\ref{tab:transitions}, the only operations that write $\mathsf{st}(n)$ are $\mathsf{mint},\mathsf{issue},\mathsf{redFlush},\mathsf{lock},\mathsf{reimburse},\mathsf{unlock}$, and their guards require $\mathsf{st}(n)\in\{\textsc{None}\}$, $\{\textsc{Blank}\}$, $\{\textsc{Issued}\}$ (for the original) or $\{\textsc{Blank}\}$ (for the credit slot $m$), $\{\textsc{Issued}\}$, $\{\textsc{Locked}\}$, and $\{\textsc{Locked}\}$ respectively. None is satisfiable when $\mathsf{st}(n)\in\{\textsc{Reimbursed},\textsc{Reversed}\}$. Hence these states are sinks.
\end{proof}

\begin{theorem}[Reimbursement uniqueness]
\label{thm:unique}
In any execution starting from $\sigma_0$, for every invoice $n$ the operation $\mathsf{reimburse}(\cdot,n)$ occurs at most once.
\end{theorem}
\begin{proof}
$\mathsf{reimburse}(\cdot,n)$ has guard $\mathsf{st}(n)=\textsc{Locked}$ and effect $\mathsf{st}(n){:=}\textsc{Reimbursed}$. Suppose it occurs at step $i$. By Lemma~\ref{lem:mono}, $\mathsf{st}(n)=\textsc{Reimbursed}$ at every step $j>i$. For a second occurrence at some $j>i$ its guard would require $\mathsf{st}(n)=\textsc{Locked}\neq\textsc{Reimbursed}$, a contradiction. Hence at most one occurrence.
\end{proof}

\begin{corollary}[Cross-organization non-duplication]
\label{cor:crossorg}
No two reimbursements of the same invoice can succeed regardless of how many distinct organizations or accounts attempt them, and regardless of concurrency.
\end{corollary}
\begin{proof}
By consensus safety the concurrent attempts are serialized into one execution (Section~\ref{sec:prelim}); Theorem~\ref{thm:unique} applies to that execution. The guarantee is a property of the single shared slot $\mathsf{st}(n)$, not of any per-organization bookkeeping, so it is independent of the callers' identities or affiliations. In a race, the first $\mathsf{lock}$ sets $\textsc{Locked}$ and every later $\mathsf{lock}$ sees $\mathsf{st}(n)\neq\textsc{Issued}$ and reverts.
\end{proof}

\begin{theorem}[Reimbursement authorization]
\label{thm:authz}
If $\mathsf{reimburse}(c,n)$ succeeds, then $c$ is the account that most recently locked $n$ and held $n$ at that time, and $n$ was not unlocked in between.
\end{theorem}
\begin{proof}
The guard requires $\mathsf{st}(n)=\textsc{Locked}$ and $L(n).\mathit{locker}=c$. $L(n)$ is written only by $\mathsf{lock}$, to $(\mathit{caller},\mathit{cid})$ under the guard $O(n)=\mathit{caller}$, and cleared only by $\mathsf{unlock}$ (which also sets $\textsc{Issued}$). Thus whenever $\mathsf{st}(n)=\textsc{Locked}$, $L(n).\mathit{locker}$ equals the caller of the most recent $\mathsf{lock}$, who held $n$ then, and no intervening $\mathsf{unlock}$ occurred (else $\mathsf{st}(n)=\textsc{Issued}$). By EUF-CMA the sender field $c$ cannot be spoofed.
\end{proof}

\begin{theorem}[Face integrity]
\label{thm:integrity}
Let $\mathsf{verify}(\phi)\triangleq H(\phi)\in\mathrm{dom}(D)$. If a PPT adversary outputs a face $\phi$ with $\mathsf{verify}(\phi)=\mathsf{true}$ and $\phi\neq F(D(H(\phi)))$, then it has computed a second preimage of $H$. Hence, under second-preimage resistance, any face that verifies equals the recorded face except with negligible probability.
\end{theorem}
\begin{proof}
$D$ is written only by $\mathsf{issue}$ and $\mathsf{redFlush}$, each as $D(H(\phi^\star)){:=}n$ for the face $\phi^\star=F(n)$ actually stored. So every $h\in\mathrm{dom}(D)$ satisfies $h=H(F(D(h)))$. If $\mathsf{verify}(\phi)$ holds then $H(\phi)=H(F(n))$ for $n=D(H(\phi))$. If additionally $\phi\neq F(n)$, then $(\phi,F(n))$ is a colliding pair with $H(\phi)=H(F(n))$; given the recorded target $F(n)$ this is a second preimage.
\end{proof}

\begin{theorem}[Red-flush value conservation]
\label{thm:redflush}
After $\mathsf{redFlush}(c,n,m)$ succeeds, the signed values $\nu(n)+\nu(m)=0$, where $\nu(k)=(1-2\,\mathsf{cr}(k))\cdot\mathsf{tot}(k)$; moreover $n$ is permanently \textsc{Reversed} and the credit $m$ can never be reimbursed or further red-flushed.
\end{theorem}
\begin{proof}
The effect copies the face of $n$ to $m$, so $\mathsf{tot}(m)=\mathsf{tot}(n)$, and sets $\mathsf{cr}(m){=}1$ while $\mathsf{cr}(n){=}0$; thus $\nu(n)=+\mathsf{tot}(n)$ and $\nu(m)=-\mathsf{tot}(n)$, summing to $0$. It sets $\mathsf{st}(n){=}\textsc{Reversed}$, terminal by Lemma~\ref{lem:mono}. For $m$: $\mathsf{lock}$ requires $\mathsf{cr}=0$ and fails on $m$; $\mathsf{redFlush}$ as an original requires $\mathsf{cr}=0$ and fails, while as a credit slot it requires the slot be \textsc{Blank}, but $\mathsf{st}(m)=\textsc{Issued}$. Hence $m$ is frozen.
\end{proof}

\subsection{Machine-Checked Verification}
\label{sec:smtchecker}
Theorems~\ref{thm:unique}--\ref{thm:redflush} are properties of the model of Section~\ref{sec:model}, whose guards and effects are transcribed directly from the contract's \texttt{require} clauses and state writes. To remove the hand-proof as a single point of trust, we additionally encode the state-machine invariants as \texttt{assert} statements in a Solidity model that preserves the registry, ownership, and every guarded transition of Table~\ref{tab:transitions} (abstracting only the invoice-face strings and hashing, which are orthogonal to these safety properties), and discharge them with the Solidity SMTChecker~\cite{smtchecker}. Its constrained-Horn-clause (CHC) engine proves properties \emph{inductively over all reachable transaction sequences and all invoice numbers}, not merely the bounded scenarios of a test suite. The ghost counter that records how often the reimburse body executes yields the direct encoding of Theorem~\ref{thm:unique}: \texttt{assert(reimburseCount[n] == 1)}. The checker reports

\begin{center}
\texttt{\footnotesize CHC: 5 verification condition(s) proved safe!}
\end{center}

\noindent covering reimbursement uniqueness (Theorem~\ref{thm:unique}), reimbursement authorization (Theorem~\ref{thm:authz}), the impossibility of reimbursing a credit or reversed invoice, and red-flush value conservation and monotonicity (Theorem~\ref{thm:redflush}, Lemma~\ref{lem:mono}), all with no counterexamples. As a non-vacuity check, weakening the \texttt{reimburse} guard to also accept an already-\textsc{Reimbursed} invoice makes the engine return a concrete counterexample trace (lock, reimburse, reimburse) that violates \texttt{reimburseCount[n]==1}, confirming the property is genuinely enforced by the guards. The 23-case test suite (Section~\ref{sec:eval}) corroborates the same properties dynamically against the deployed bytecode.

Goals G1--G3 of Section~\ref{sec:threat} are thus refuted by Theorem~\ref{thm:unique}/Corollary~\ref{cor:crossorg}, Theorem~\ref{thm:integrity}, and Theorems~\ref{thm:authz}/\ref{thm:redflush} together with the per-operation guards, respectively, with the state-machine goals additionally machine-checked.

\section{Implementation}
\label{sec:impl}
The system is implemented in three parts (a fourth component, the autonomous agent runtime, is described in Section~\ref{sec:agents}). \textbf{Contract:} \texttt{EInvoice.sol} is a single Solidity~0.8.24 contract (487 lines) compiled with the optimizer and the IR pipeline~\cite{solidity}. Checked arithmetic removes the overflow bug class; the contract has no external calls, no Ether handling, and no unbounded loops in state-changing paths except range minting, whose bound the bureau controls. All transitions emit indexed events; following~\cite{wohrer2018}, access control is expressed as guard-first \texttt{require} clauses, and with no external calls reentrancy is structurally excluded. \textbf{Tests and benchmarks:} a 23-case Hardhat~\cite{hardhat} suite covers every subsystem---happy paths, each access-control rejection, double-issue, double-declare, double-reimbursement, unauthorized red-flush, unlock semantics, forged-digest rejection, and a full lifecycle integration test---and deploy/demo/benchmark scripts reproduce every number in Section~\ref{sec:eval}. A separate SMTChecker model and a one-command verification script (Section~\ref{sec:smtchecker}) reproduce the machine-checked proofs. \textbf{Web application:} a React/ethers.js app provides four dashboards (bureau, seller, buyer, verifier); the buyer sees received invoices rendered as VAT-style faces, verifies them by recomputed digest, and drives the lock--reimburse--unlock protocol, while the verifier reconstructs an invoice's full audit trail from events. Contract revert reasons surface in the interface, so a duplicate reimbursement attempt visibly fails with the contract's own message.

\section{Evaluation}
\label{sec:eval}

\subsection{Setup}
Experiments run on a commodity laptop (Intel Core i9-14900HX, 32~GB RAM, Linux~6.17) against Hardhat's in-process EVM and, for the end-to-end demonstration, a local JSON-RPC node. Gas costs are deterministic properties of the EVM and transfer unchanged to any Ethereum-compatible deployment; latency and throughput are properties of the single-node configuration and are read as \emph{contract-execution} bounds, with consensus overhead added in a real network.

\subsection{Functional Correctness}
All 23 tests pass, including the adversarial cases: an unregistered account cannot apply, mint, issue, or receive; a non-holder cannot issue or lock; a non-seller cannot declare or red-flush; a reimbursed invoice cannot be locked, reimbursed again, or red-flushed; a forged face fails digest verification. The complete lifecycle---apply, approve, issue, third-party verify, declare, lock, reimburse---was additionally exercised end-to-end through the web application against the live local chain, including the visible rejection of a duplicate reimbursement attempt.

\subsection{Gas Cost and Complexity}
Table~\ref{tab:gas} reports per-operation gas over the benchmark workload (200 invoices; 50 red-flushes) alongside asymptotic time and storage complexity in the number of invoices $N$; Fig.~\ref{fig:gas} visualizes the costs. One-time deployment costs 3{,}861{,}367~gas.

\begin{table}[t]
\caption{Per-operation gas cost and complexity ($N$ invoices, batch size $k$).}
\label{tab:gas}
\centering
\renewcommand{\arraystretch}{1.12}
\begin{tabular}{@{}l r c c@{}}
\toprule
Operation & Gas (mean) & Time & Storage \\
\midrule
registerEnterprise & 189{,}326 & $O(1)$ & $O(1)$ \\
applyForInvoices & 92{,}918 & $O(1)$ & $O(1)$ \\
grantInvoices (per blank) & 140{,}731 & $O(k)$ & $O(k)$ \\
issueInvoice & 646{,}773 & $O(1)$ & $O(1)$ \\
redFlush & 696{,}130 & $O(1)$ & $O(1)$ \\
declareTax & 32{,}716 & $O(1)$ & $O(1)$ \\
lockForReimbursement & 101{,}902 & $O(1)$ & $O(1)$ \\
reimburse & 32{,}743 & $O(1)$ & $O(1)$ \\
unlockReimbursement & 35{,}031 & $O(1)$ & $O(1)$ \\
\bottomrule
\end{tabular}
\end{table}

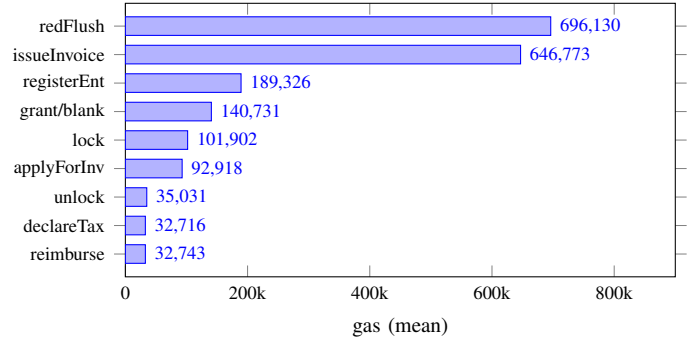
\begin{figure}[t]
\centering
\begin{tikzpicture}
\begin{axis}[
  xbar, width=\columnwidth, height=5.2cm,
  xmin=0, xmax=900000,
  scaled x ticks=false,
  bar width=7pt,
  enlarge y limits=0.10,
  xlabel={\footnotesize gas (mean)},
  symbolic y coords={reimburse,declareTax,unlock,applyForInv,lock,grant/blank,registerEnt,issueInvoice,redFlush},
  ytick=data,
  yticklabel style={font=\scriptsize},
  xticklabel style={font=\scriptsize},
  xtick={0,200000,400000,600000,800000},
  xticklabels={0,200k,400k,600k,800k},
  nodes near coords={\scriptsize\pgfmathprintnumber[fixed,precision=0,1000 sep={,}]{\pgfplotspointmeta}},
  nodes near coords style={font=\tiny, /pgf/number format/assume math mode=true},
  every axis plot/.append style={fill=barfill!75, draw=barfill},
]
\addplot coordinates {
  (32743,reimburse) (32716,declareTax) (35031,unlock)
  (92918,applyForInv) (101902,lock) (140731,grant/blank)
  (189326,registerEnt) (646773,issueInvoice) (696130,redFlush) };
\end{axis}
\end{tikzpicture}
\caption{Per-operation gas cost. Issuance and red-flush dominate because they snapshot a full invoice face into permanent storage; the anti-fraud lock/reimburse pair is an order of magnitude cheaper.}
\label{fig:gas}
\end{figure}

Issuance and red-flush dominate because they snapshot both parties' registry data and the item description into permanent storage---the price of a self-contained, immutable face. Crucially, every core operation is $O(1)$ in $N$: costs do not grow as the ledger fills, so the design scales to national invoice volumes at the contract level; only batch granting is linear in its (bureau-chosen) batch size $k$. The anti-fraud protocol is cheap: lock plus reimburse total under 135{,}000~gas, about a fifth of issuance. On a consortium chain gas is only a resource meter; on public Ethereum at 1~gwei and \$3{,}000/ETH, issuance would cost about \$1.94 and reimbursement about \$0.40, motivating consortium or layer-2 deployment for high-volume use.

\subsection{Latency and Throughput}
Against the in-process node, mean end-to-end transaction latency (submission to mined receipt) is 1.6--3.2~ms per operation ($p_{95}\le 8.2$~ms), confirming that contract execution is negligible next to consensus in any realistic deployment. Submitting 100 concurrent issuance transactions completes in 0.73~s, i.e., about 137~invoices/s sustained on a single development node. For scale, China's Shenzhen pilot averaged roughly 0.2~invoices/s in its first year~\cite{cointelegraph2019}, and a consortium deployment can shard by region or bureau, so contract-level throughput is not the binding constraint; consensus configuration is.

\section{Autonomous Agents over a Verified Invoice Ledger}
\label{sec:agents}

Invoice matching, expense reimbursement, and tax preparation are among the first enterprise workflows being delegated to LLM-based agents~\cite{taxagent2025}, and agents that hold keys and transact on blockchains are an active research area whose central open problem is \emph{verifiable policy enforcement}: how to bound what an autonomous, possibly misbehaving agent can do with delegated authority~\cite{alqithami2026,gong2026}. Recent guardrail systems synthesize and verify bespoke policies around the agent~\cite{veriguard2025,shieldagent2025}. Our setting admits a simpler answer: \textbf{the guardrail already exists and is already verified --- it is the tax authority's own contract.} An agent participates as a registered account, which is precisely the adversary of the threat model (Section~\ref{sec:threat}); no additional mechanism, and no trust in the LLM, is needed for the safety properties to hold.

\begin{corollary}[Agent safety]
\label{cor:agent}
Let an autonomous agent control the keys of one or more registered accounts and issue an arbitrary sequence of contract calls (arising, e.g., from LLM outputs, software faults, or prompt injection). Then no behavior of the agent can (i) cause any invoice to be reimbursed more than once, (ii) cause a face to verify that differs from the recorded one, or (iii) effect a lifecycle transition the controlled accounts are not authorized to perform.
\end{corollary}
\begin{proof}
The agent's observable effect on the ledger is a set of signed transactions from registered accounts, i.e., an instance of $\mathcal{A}$ in Section~\ref{sec:threat}. Claims (i)--(iii) are Theorem~\ref{thm:unique}/Corollary~\ref{cor:crossorg}, Theorem~\ref{thm:integrity}, and Theorem~\ref{thm:authz} with the per-operation guards, respectively. The machine-checked result of Section~\ref{sec:smtchecker} quantifies over \emph{all} transaction sequences, hence over all agent behaviors.
\end{proof}

Prompt injection deserves one remark: an attacker who fully controls the agent's inputs (claim texts, attached receipts) controls at most which transactions the agent's account submits --- capabilities the account holder already possesses legitimately. Injection can therefore waste gas or misdirect effort, but cannot cross the contract's authorization or uniqueness boundaries. What the contract does \emph{not} provide is liveness or decision quality: a broken agent can fail to reimburse valid claims or match a claim to the wrong (still-authorized) invoice, which is why we retain a policy layer and a human escalation path.

\subsection{Agent Architecture and Implementation}

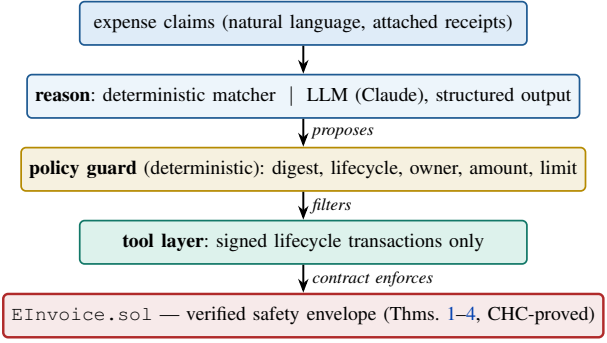
\begin{figure}[t]
\centering
\resizebox{0.9\columnwidth}{!}{%
\begin{tikzpicture}[
  font=\footnotesize,
  lay/.style={draw, rounded corners=2pt, align=center, minimum width=6.4cm, minimum height=0.62cm, line width=0.7pt},
  e/.style={-{Stealth[length=2mm]}, line width=0.8pt}
]
\node[lay, draw=appline, fill=appfill] (claims) {expense claims (natural language, attached receipts)};
\node[lay, draw=appline, fill=appfill!55, below=4mm of claims] (reason) {\textbf{reason}: deterministic matcher $\;|\;$ LLM (Claude), structured output};
\node[lay, draw=gold, fill=gold!12, below=4mm of reason] (guard) {\textbf{policy guard} (deterministic): digest, lifecycle, owner, amount, limit};
\node[lay, draw=ctrline, fill=ctrfill, below=4mm of guard] (tools) {\textbf{tool layer}: signed lifecycle transactions only};
\node[lay, draw=acc, line width=1pt, fill=acc!10, below=4mm of tools] (contract) {\texttt{EInvoice.sol} --- verified safety envelope (Thms.~\ref{thm:unique}--\ref{thm:redflush}, CHC-proved)};
\draw[e] (claims) -- (reason);
\draw[e] (reason) -- node[right, font=\scriptsize\itshape]{proposes} (guard);
\draw[e] (guard) -- node[right, font=\scriptsize\itshape]{filters} (tools);
\draw[e] (tools) -- node[right, font=\scriptsize\itshape]{contract enforces} (contract);
\end{tikzpicture}%
}
\caption{Agent stack. The reasoner (possibly an LLM) only \emph{proposes}; a deterministic policy filters; the verified contract is the final, un-bypassable authority.}
\label{fig:agent}
\end{figure}

We implement two agents in a small Node.js runtime (Fig.~\ref{fig:agent}). A \emph{reimbursement agent} runs a perceive--reason--guard--act loop over a queue of natural-language expense claims: it reads its account's invoice holdings, matches each claim to an invoice, applies a deterministic policy (the attached receipt must verify by digest against the same invoice; the invoice must be \emph{Issued}, non-credit, and held by the agent; the claim amount must equal the invoice total; a per-claim spending ceiling applies), and only then executes the on-chain lock--reimburse protocol, recording a structured trace. Reasoning is \emph{hybrid}: a deterministic matcher (exact amount plus description-token overlap) by default, or an LLM (Claude, with schema-constrained JSON output and a validity check that discards any invoice identifier not in the candidate set) when an API key is configured, falling back to rules on any error. The architecture makes the trust relationship explicit: the reasoner proposes, the policy filters, and the contract enforces. An \emph{audit agent} independently scans the entire ledger, re-verifying every face digest, the tax arithmetic, and red-flush linkage, and flagging identical-content invoice pairs for human review.

\subsection{Case Study}

\begin{table}[t]
\caption{Reimbursement-agent case study (deterministic reasoner; six claims over five seeded invoices; local chain).}
\label{tab:agent}
\centering
\renewcommand{\arraystretch}{1.15}
\resizebox{\columnwidth}{!}{%
\begin{tabular}{@{}l l l l@{}}
\toprule
Claim & Nature & Decided by & Outcome \\
\midrule
C1--C3 & legitimate & contract & reimbursed on-chain \\
C4 & duplicate of C1 & policy & rejected: ``Reimbursed, not Issued'' \\
C5 & over spending limit & policy & rejected: exceeds ¥5{,}000 ceiling \\
C6 & forged receipt (amounts $\times 2$) & policy & rejected: digest verification fails \\
C4$'$ & duplicate, \emph{policy bypassed} & \textbf{contract} & reverted: ``not available for reimbursement'' \\
\midrule
\multicolumn{4}{@{}l}{Unsafe actions executed on-chain: \textbf{0 of 4} attempts. Audit: 6 invoices, 0 findings.}\\
\bottomrule
\end{tabular}%
}
\end{table}

Table~\ref{tab:agent} summarizes an end-to-end run. All three legitimate claims were matched and reimbursed; the duplicate, over-limit, and forged claims were rejected by the policy with human-readable reasons; and when we deliberately replaced the policy with one that approves everything, the duplicate reimbursement was rejected \emph{by the contract itself}, exactly as Corollary~\ref{cor:agent} predicts. Re-running with the LLM reasoner can change only which invoice a claim is matched to --- never the safety outcomes --- because every unsafe action is rejected downstream of the reasoner. The run is scripted and reproducible (\texttt{agent/demo.js}), and twelve additional offline tests exercise the reasoner, the policy, both agents, and the policy-bypass path.

\section{Discussion and Limitations}
\label{sec:discussion}
\textbf{Privacy.} Invoice faces are stored in plaintext contract storage, acceptable on a permissioned consortium chain with restricted read access but not on a public chain. The digest-verification design already points to the remedy: store only commitments on chain and keep faces off chain, or apply selective encryption; zero-knowledge proofs could further allow verifying properties (e.g., ``total below limit, not yet reimbursed'') without revealing the face.

\textbf{Transaction authenticity.} Our guarantees cover invoice \emph{integrity and lifecycle}, not the reality of the underlying sale---a seller can still invoice a fictitious transaction (goal outside Section~\ref{sec:threat}). Anchoring digests of orders, logistics, and payment records and requiring their presence at issuance would raise the cost of fictitious invoicing; integration with payment systems (as in the Shenzhen deployment~\cite{tencent2018}) closes this gap in practice.

\textbf{Deployment path.} The contract runs unmodified on public, private, and consortium EVM chains. The realistic production path is a consortium chain operated by tax authorities with enterprise nodes, where finality is fast and gas is unpriced; a layer-2 rollup anchored to a public chain is an alternative offering public verifiability at lower cost. Storage growth (roughly 1--2~KB per face) favors the commitment-on-chain variant at national volumes.

\textbf{Legal integration.} Invoice numbering, red-flush semantics, and declaration follow Chinese VAT practice; the state machine itself is jurisdiction-neutral, and the EU's structured e-invoicing mandate~\cite{vida2025} makes the lifecycle-on-ledger approach broadly relevant.

\section{Conclusion}
\label{sec:conclusion}
We presented a complete, formally-grounded blockchain electronic invoice system on Ethereum: a five-subsystem design covering the full invoice lifecycle, a formal transition-system model, and machine-checkable security arguments proving reimbursement uniqueness (even across organizations), face integrity, red-flush value conservation, and authorization soundness. The design models invoices as non-tradable NFTs governed by an explicit automaton, with a lock-based reimbursement protocol that makes duplicate reimbursement unrepresentable and digest-based verification that makes authenticity checking trustless. The open-source implementation passes an adversarial test suite; measurements show modest, $O(1)$ per-operation costs (646{,}773~gas to issue; $<$135{,}000~gas for the full reimbursement protocol) and throughput far above production pilot volumes. The two pain points that most impede e-invoice adoption on the consumption side---duplicate reimbursement and costly verification---are eliminated by construction. Beyond human users, we showed the verified contract serves as a ready-made safety envelope for autonomous AI agents: because an agent is just a registered account, agent safety follows as a corollary of the same theorems, and our agentic case study confirms that duplicate, forged, and over-limit claims are rejected even when the agent's own safeguards are bypassed. Future work targets the privacy layer (commitments and zero-knowledge verification), anchoring of transaction evidence at issuance, richer agent autonomy (issuance- and declaration-side agents, multi-agent negotiation), and a multi-node consortium evaluation with PBFT-class consensus.

\section*{Artifact Availability}
The smart contract, SMTChecker verification model, four-role web application, autonomous-agent runtime, test suite (23 contract cases plus 12 agent cases), and the deploy, demonstration, and benchmark scripts that reproduce every number reported here are available under the MIT license at \url{\artifacturl}. The one-command scripts \texttt{verification/verify.sh} and \texttt{agent/demo.js} reproduce, respectively, the machine-checked result of Section~\ref{sec:smtchecker} and the agent case study of Table~\ref{tab:agent}.

\bibliographystyle{IEEEtran}
\bibliography{refs}

\end{document}